\documentclass[a4paper, 11pt]{article}
\usepackage{amsthm,amssymb,amsmath,enumerate,graphicx}
\usepackage{color}

\usepackage{authblk}

\newtheorem{definition}{Definition}
\newtheorem{proposition}[definition]{Proposition}
\newtheorem{theorem}[definition]{Theorem}

\newtheorem{lemma}[definition]{Lemma}

\newtheorem{claim}[definition]{Claim}

\newcommand{\comment}[1]{}
\newcommand{\dom}{\hbox{\rm dom}}

\title{$b$-coloring is NP-hard on co-bipartite graphs and polytime solvable on tree-cographs \thanks{Partially supported by UBACyT Grant 20020100100980 and CONICET
PIP 112-200901-00178 (Argentina) and MathAmSud Project 13MATH-07
(Argentina--Brazil--Chile--France).}}

\author[1]{Flavia Bonomo}
\author[2]{Oliver Schaudt}
\author[3]{Maya Stein}
\author[4]{\\Mario Valencia-Pabon\thanks{Actually in "D\'el\'egation" at INRIA Nancy - Grand Est, 2013-2014.}}

\affil[1]{\small CONICET and Departamento de Computaci\'on,
Facultad de Ciencias Exactas y Naturales, Universidad de Buenos
Aires, Buenos Aires, Argentina. Email:~fbonomo@dc.uba.ar}
\affil[2]{\small Institut de Math\'ematiques de Jussieu, CNRS
UMR7586, Universit\'e Pierre et Marie Curie (Paris 6), Paris,
France. Email:~schaudt@math.jussieu.fr} 
\affil[3]{\small Centro de
Modelamiento Matem\'atico, Universidad de Chile, Santiago, Chile.
Email:~mstein@dim.uchile.cl} 
\affil[4]{\small Universit\'e Paris
13, Sorbonne Paris Cit\'e, LIPN, CNRS UMR7030, Villetaneuse, France. Email:~valencia@lipn.univ-paris13.fr}

\date{}

\begin{document}
\maketitle

\begin{abstract}
A \emph{b-coloring} of a graph is a proper coloring such that
every color class contains a vertex that is adjacent to all other
color classes. The \emph{b-chromatic number} of a graph $G$,
denoted by $\chi_b(G)$, is the maximum number $t$ such that $G$
admits a b-coloring with $t$ colors. A graph~$G$ is called
\emph{b-continuous} if it admits a b-coloring with $t$ colors, for
every $t = \chi(G),\ldots,\chi_b(G)$, and \emph{b-monotonic} if
$\chi_b(H_1) \geq \chi_b(H_2)$ for every induced subgraph $H_1$ of
$G$, and every induced subgraph $H_2$ of $H_1$.

We investigate the b-chromatic number of graphs with stability number two. 
These are exactly the complements of triangle-free graphs, 
thus including all complements of bipartite graphs. 
The main results of this work are the following:
\begin{enumerate}
\item We characterize the b-colorings of a graph with stability
number two in terms of matchings with no augmenting paths of
length one or three. 
We derive that graphs with stability number two are b-continuous and b-monotonic.
\item We prove that it is NP-complete to decide whether the b-chromatic
number of co-bipartite graph is at most a given threshold. 
\item 
We describe a polynomial
time dynamic programming algorithm to compute the b-chromatic
number of co-trees. 
\item Extending several previous results, we show that there is a polynomial time
dynamic programming algorithm for computing the b-chromatic number
of tree-cographs.
Moreover, we show that tree-cographs are b-continuous and
b-monotonic.
\end{enumerate}
\end{abstract}


\section{Introduction}

A \emph{b-coloring} of a graph $G$ by $k$ colors is a proper
coloring of the vertices of $G$ such that every color class
contains a vertex that is adjacent to all the other $k-1$ color
classes. Such a vertex will be called a \emph{dominating vertex}.
It is easy to see that any coloring of  a graph $G$ with $\chi(G)$
many colors  is a b-coloring (as usual, we denote by $\chi (G)$
the minimum number of colors needed for a proper coloring of the
vertices of a graph).

The \emph{b-chromatic number} of a graph $G$, denoted by
$\chi_b(G)$, is the maximum number $k$ such that $G$ admits a
b-coloring with $k$ colors. Clearly, $\chi_b(G) \leq \Delta(G) +1$
where $\Delta(G)$ denotes the maximum degree of $G$. The
b-chromatic number was introduced in \cite{I-M-b-col}. The
motivation, similarly as the well known achromatic number (cf.
e.g, \cite{Harary-Hedetniemi70, Bodlaender89} and ref. therein),
comes from algorithmic graph theory. Suppose one colors a given
graph properly, but in an arbitrary way. After all vertices are
colored, one would wish to perform some simple operations to
reduce the number of colors. A simple operation consists in
recoloring all the vertices in one color class with a possible
different color. Obviously, such recoloring is impossible if each
color class contains a dominating vertex. Hence, the b-chromatic
number of the graph serves as the tight upper bound for the number
of colors used by this coloring heuristic. From this point of
view, both complexity results and polynomial time algorithms for
particular graph families are interesting.

Assume that the vertices $v_1,v_2,\ldots,v_n$ of a graph $G$ are
ordered such that $d(v_1) \geq d(v_2) \geq \ldots \geq d(v_n)$,
where $d(x)$ denotes the degree of vertex $x$ in $G$. Let $$m(G)
:= \max\{i : d(v_i) \geq i-1\}$$ be the maximum number $i$ such
that $G$ contains at least $i$ vertices of degree $\geq i-1$. It
is clear that $m(G) \leq \Delta(G) +1$. Irving and Manlove
\cite{I-M-b-col} show that this parameter bounds the b-chromatic
number:

\begin{proposition}
For every graph $G$, $\chi(G) \leq \chi_b(G) \leq m(G)$.
\end{proposition}
 Irving and Manlove \cite{I-M-b-col} also show that determining $\chi_b(G)$ is NP-complete for general graphs, but polynomial-time solvable for trees. Kratochv\'{\i}l, Tuza and Voigt \cite{K-T-V-b-col} prove that the problem of determining if $\chi_b(G) = m(G)$ is NP-complete even for connected bipartite graphs $G$ with $m(G) = \Delta(G) +1$. A graph $G$ is \emph{tight} if it has exactly $m(G)$ \emph{dense} vertices (a vertex $v$ of a graph $G$ is dense if $d(v) \geq m(G)-1$), each of which has degree exactly $m(G)-1$.  Havet, Linhares-Sales and Sampaio \cite{Havet-et-al12} recently investigated the problem on tight graphs. They proved that the problem of determining if a tight graph $G$ has $\chi_b(G) = m(G)$ is NP-complete for bipartite graphs and ptolemaic graphs, but polynomial-time solvable for complements of bipartite graphs, split graphs and block graphs.

In last years, several related concepts concerning b-colorings of
graphs have been studied in
\cite{Faik-tesis,Havet-et-al12,H-K-bperf,H-L-M-b-imperf,K-K-V-b-col}. A graph
$G$ is defined to be \emph{b-continuous} \cite{Faik-tesis} if it
admits a b-coloring with $t$ colors, for every $t=\chi(G), \dots,
\chi_b(G)$.  In \cite{K-K-V-b-col} (see also \cite{Faik-tesis}) it
is proved that chordal graphs and some planar graphs are
b-continuous. A graph $G$ is defined to be \emph{b-monotonic}
\cite{Bonomo-et-al09} if $\chi_b(H_1) \geq \chi_b(H_2)$ for every
induced subgraph $H_1$ of $G$, and every induced subgraph $H_2$ of
$H_1$. They prove  that $P_4$-sparse graphs (and, in particular,
cographs) are b-continuous and b-monotonic. Besides, they give a
dynamic programming algorithm to compute the b-chromatic number in
polynomial time within these graph classes.\\

Our paper is organized as follows. In the next section, we
characterize  b-colorings of  graphs with stability number two in
terms of matchings with no augmenting paths of length one or
three.

In Section 3, we prove that graphs with stability at most two are
both b-continuous and b-monotonic.

In Section 4, we prove that computing the b-chromatic number of
co-bipartite graphs is an NP-complete problem.

Finally, in Section 5,  first we describe a polynomial-time
dynamic programming algorithm to compute the b-chromatic number of
co-trees. Next, we extend our results to the family of
tree-cographs by showing that there is a polynomial time dynamic
programming algorithm for computing the b-chromatic number of
graphs in this family and that these are also b-continuous and
b-monotonic.


\section{b-colorings and matchings}

The \emph{stability} of a graph $G$ is defined as the maximum
cardinality of a subset of pairwise non-adjacent vertices in $G$.
Given a graph $G$, we denote by $\overline{G}$ the complement
graph of $G$, which is the graph on the same set of vertices as
$G$ that has an edge between two different vertices $u$ and $v$ if
and only if $u$ and $v$ are non-adjacent in $G$. It is not
difficult to see that $G$ is a graph with stability one if and
only if it is complete, and $G$ is a graph with stability at most
two if and only if $\overline{G}$ is a triangle-free graph. In
this section, we will see that matchings in triangle-free graphs
are very important when we deal with b-colorings of graphs with
stability at most two.

Let $M$ be a matching of a graph $G$. Denote by $V(M)$ the set of
all vertices covered by $M$. An \emph{augmenting path} for $M$ is
a path starting and ending outside $V(M)$ whose edges alternate
between $E(G)-M$ and $M$. Usually, $M$ is called \emph{maximal} if
no further edge can be included in $M$. In other words, $G$ does
not contain an augmenting path of length one with respect to $M$.
Following this terminology we call $M$ \emph{strongly maximal} if
$G$ does not contain augmenting paths of length one or three with
respect to $M$. Trivially, maximum matchings are strongly maximal,
and strongly maximal matchings are maximal. Our next lemma shows
why strongly maximal matchings are important in our setting.

\begin{lemma}\label{lem:matching}
Let $G$ be a graph of stability at most two and let $c$ be a
proper coloring of $G$.  Then $c$ is a $b$-coloring if and only if
the set
\[
M = \{uv : u,v \in V, u \neq v \mbox{ and } c(u)=c(v)\}
\]
is a strongly maximal matching in $\overline{G}$. Moreover, the
number of colors $c$ uses is $|V(G)|-|M|$.
\end{lemma}

\begin{proof}
First, observe that $M$ is a (possibly empty) matching of
$\overline{G}$ because $G$ has stability at most two. Now, suppose
that $\overline{G}$ contains an augmenting path $P$ of length $1$
or~$3$ for~$M$. If $P$ consists of only one edge $uv$, then in
$G$, the vertices $u$ and $v$ are non-adjacent, and each makes up
a singleton colour class. Thus $c$ is not a $b$-coloring. If $P$
has three edges, then for each of the endvertices of its middle
edge $uv$ there is a singleton color class which it does not see
in $G$. So the color class $\{u,v\}$ witnesses the fact that $c$
is not a $b$-coloring.

Next, suppose that $c$ is not a $b$-coloring. Note that, as $G$
has stability at most two, every vertex of $G$ is adjacent (in
$G$) to at least one vertex of any given color class of size~$2$.
So, the witness for $c$ not being a $b$-coloring is one of the
following two: either it is a singleton color class whose vertex
is non-adjacent to another singleton color class, or it is a color
class  $\{u,v\}$ of size two, such that  $u$ is non-adjacent to
some singleton color class, and $v$ is non-adjacent to a different
singleton color class. Clearly, the first situation corresponds to
an augmenting path of $M$ on one edge, and the second  situation
corresponds to an augmenting path of $M$ on three edges.
\end{proof}

Observe that coloring $c$ from Lemma~\ref{lem:matching} is a
maximum (minimum) $b$-coloring of $G$ if and only if $M$ is a
minimum (maximum) strongly maximal matching of $\overline{G}$.


\section{b-continuity and b-monotonicity of graphs with stability at most two}\label{sec:b-con_b-mon}

In order to prove the b-continuity of graphs with stability at
most two, we need the following result.

\begin{lemma}
\label{le-b-cont-2}
Let $M$ be a strongly maximal matching of a 
graph $G$ and let $P$ be a minimum length augmenting path in $G$
with respect to $M$. Then, the matching $M' = (M \setminus E(P))
\cup (E(P) \setminus M)$ is a strongly maximal matching of $G$,
and $|M'| = |M|+1$.
\end{lemma}

\begin{proof}
Let $P=(x_1,x_2,\ldots,x_k)$. By basic results from matching
theory, the only thing we need to prove is that
 $M'$ is again strongly maximal.
Since the maximality of $M'$ is clear, suppose for contradiction
that there is an augmenting path of length 3, say $Q=(u,v,w,x)$.
Necessarily $vw$ is an edge of $M' \setminus M$, and thus
w.l.o.g.~there is some $i \in \{1,2,\ldots,k-1\}$ with $v=x_i$ and
$w=x_{i+1}$. Moreover, $u,x \notin V(M)$. Thus both paths
$(x_1,x_2,\ldots,x_i,u)$ and $(x,x_{i+1},x_{i+2},\ldots,x_k)$ are
augmenting paths for $M$ and at least one of these paths is
shorter that $P$. This is a contradiction to the choice of $P$.
\end{proof}

By Lemma \ref{lem:matching}, any b-coloring using $k > \chi(G)$
colors of  a graph $G$ of stability at most two corresponds to a
strongly maximal matching $M$ that is not maximum. By Berge's
lemma~\cite{Berge57}, there is an augmenting path for $M$. Using
Lemma~\ref{le-b-cont-2} we obtain a strongly maximal matching $M'$
of cardinality $|M|+1$, which, again by Lemma~\ref{lem:matching},
corresponds to a b-coloring with $k-1$ colors. Repeatedly applying
this argument gives the following result.

\begin{theorem}\label{thm:bconti}
Graphs of stability at most two are $b$-continuous.
\end{theorem}

Given a maximum b-coloring  of  a graph $G$ of stability at most
two, we can thus find b-colorings for all values between $\chi(G)$
and $\chi_b(G)$. Moreover, we can do this in polynomial time,
provided we can find a minimum length augmenting path for a given
matching in polynomial time. This is the aim of the following
lemma that can be derived by a slight modification of Edmonds' blossom algorithm \cite{edmonds_match}.

\begin{lemma}
\label{le-b-cont-1} Let $M$ be a  matching in a graph $G$. Then, a
minimum length augmenting path $P$ in $G$ with respect to $M$ can
be computed in polynomial time.
\end{lemma}

Lemma \ref{le-b-cont-1} together with the proof of
Theorem~\ref{thm:bconti} implies that given a graph $G$ of
stability at most two, and a b-coloring of $G$ using $k > \chi(G)$
colors, we can compute in polynomial time a b-coloring for $G$
with $k-1$ colors. Notice that the converse is not necessarily
true, i.e., if we have a b-coloring of $G$ using $k < \chi_b(G)$
colors, we do not know how to compute in polynomial time a
b-coloring for $G$ with $k+1$ colors. Indeed, we will prove in the
next section that the problem of computing the b-chromatic number
of a graph with stability at most two is NP-complete, even
restricted to the smaller class of co-bipartite graphs.

\medskip

We now turn to the $b$-monotonicity of graphs of stability at most two.

\begin{theorem}
Graphs of stability at most two are $b$-monotonic.
\end{theorem}

\begin{proof}
The class of graphs of stability at most two is closed under taking induced subgraphs.
Thus we only have to prove that $\chi_b$ is monotonously decreasing under the deletion of a vertex.
In view of Lemma~\ref{lem:matching}, it is sufficient to show that given a graph $G$ of stability $2$ and some vertex $v \in V(G)$ the following holds:
If there is a strongly maximal matching of $\overline{G-v}$ of size $k$, then there is a strongly maximal matching of $\overline{G}$ of size at most $k+1$.
This implies $\chi_b(G) \ge \chi_b(G-v)$.

Let $M$ be a strongly maximal matching of $\overline{G-v}$.
We now consider $M$ as a matching of $\overline{G}$.
If $M$ is a strongly maximal matching of $\overline{G}$, we are done.
So we assume that $M$ is not strongly maximal, and thus there is an augmenting path $P$ of length at most 3.
We may choose $P$ of minimum length among all augmenting paths of $M$ in $\overline{G}$. Note that $P$ meets $v$, say $P$ starts in $v$.

We assume first that $P$ is of length 1, that is, $v$ has an unmatched neighbor in $\overline{G}$, say $u$.
We claim the matching $M' = M \cup \{uv\}$ is strongly maximal, in which case we are done. Indeed, otherwise there is an augmenting path of length 3 for $M'$, and $uv$ is the central edge of this path. So, there is a neighbor of $u$, say $w$, that is not matched by $M'$.
In particular, $v \neq w$.
Thus $uw$ is disjoint from $M$, contradicting the maximality of $M$ in $\overline{G-v}$.
This  proves our claim.

Now assume that $P$ is of length 3, say $P = (v,u,w,x)$.
Let $M' = (M \setminus \{uw\}) \cup \{vu,wx\}$.
Suppose that $M'$ is not strongly maximal in $\overline{G}$.
Then there is an augmenting path of length 3, with central edge either $vu$ or $wx$.
In either case, $x$ or $v$ has a neighbor in $\overline{G}$ that is unmatched by $M'$ and thus also by $M$, a contradiction either to the maximality of $M$ in $\overline{G-v}$ or to the minimality of $P$.
This completes the proof.
\end{proof}


\section{NP-hardness result for co-bipartite graphs}

As mentioned in Section 1, Havet, Linhares-Sales and Sampaio \cite{Havet-et-al12} proved that the problem of determining if a tight co-bipartite graph $G$ has $\chi_b(G) = m(G)$ is polynomial-time solvable. 
However, the computational complexity of $\chi_b$ in the class of co-bipartite graphs is left open. 
In the next theorem, we prove that b-coloring general co-bipartite graphs is a hard problem.

\begin{theorem}\label{thm:hardness}
Given a co-bipartite graph $G$ and a natural number $k$, it is NP-complete to decide whether $G$ admits a $b$-coloring with at least $k$ colors.
\end{theorem}

\begin{proof}
By Lemma~\ref{lem:matching}, it suffices to prove that it is NP-complete to decide whether a bipartite graph $G$ admits a strongly maximal matching containing at most $k$ edges, when $G$ and $k$ are given input.

Our reduction is from the minimum maximal matching problem which is to decide whether a given graph admits a maximal matching of at most $k$ edges, for given $k$.
This problem is NP-complete even if the instances are restricted to bipartite graphs, as shown by Yannakakis and Gavril~\cite{Yannakakis-Gavril80}.

Given a bipartite graph $G$ with $m$ edges, we define a new graph $H_G$ as follows.
For each edge $uv \in E(G)$ we introduce a set of new vertices
\[
X_{uv}=\{x^1_{uv},x^2_{uv},x^3_{uv},x^4_{uv},x^1_{vu},x^2_{vu},x^3_{vu},x^4_{vu}\}
\]
and edges
\[
F_{uv}=\{ux^1_{uv},x^1_{uv}x^2_{uv},x^2_{uv}x^3_{uv},x^3_{uv}x^4_{uv},x^1_{uv}x^1_{vu},x^1_{vu}x^2_{vu},x^2_{vu}x^3_{vu},x^3_{vu}x^4_{vu},vx^1_{vu}\}.
\]
Note that $X_{uv}=X_{vu}$ and $F_{uv}=F_{vu}$.
Then $H_G$ is defined by
\begin{align*}
V(H_G) & = V(G) \cup \bigcup_{uv \in E(G)} X_{uv},\\
E(H_G) & = \bigcup_{uv \in E(G)} F_{uv}.
\end{align*}
Clearly $H_G$ can be computed in polynomial time.
Moreover, $H_G$ is bipartite since $G$ is. For each edge $uv \in E(G)$, we define the following auxiliary sets of edges in $H_G$:
\[
F_{uv}^{\in}=\{ux^1_{uv},x^2_{uv}x^3_{uv},x^2_{vu}x^3_{vu},vx^1_{vu}\}\text{ and }F_{uv}^{\notin}=\{ x^1_{uv}x^1_{vu},x^2_{uv}x^3_{uv},x^2_{vu}x^3_{vu}\}
\]
We claim the following:
\begin{claim}
\label{claimNP}
There exists a minimum strongly maximal matching $M$ of $H_G$ such that
\begin{equation*}\label{motte}
\text{$x_{uv}^3x_{uv}^4\notin M$ for each edge $uv\in E(G)$.}
\end{equation*}
Moreover, $M$ can be obtained from any minimum strongly maximal matching of $H_G$ in polynomial time.
\end{claim}
In order to prove this claim, we proceed by contradiction. Assume
that every minimum strongly maximal matching of $H_G$ contains at
least an edge $x_{uv}^3x_{uv}^4$ for some edge $uv\in E(G)$, and
let $M$ be a minimum strongly maximal matching of $H_G$ having a
minimum number of edges of the form $x_{uv}^3x_{uv}^4$. Note that
the choice of $M$ implies that for every edge $uv \in E(G)$ we
have that
\begin{itemize}
\item[$(i)$] $x_{uv}^3x_{uv}^4\in M$ if and only if
$x_{uv}^1x_{uv}^2\in M$. If $x_{uv}^1x_{uv}^2\in M$ then
$x_{uv}^3x_{uv}^4\in M$, otherwise, $M$ is not maximal. If
$x_{uv}^3x_{uv}^4\in M$ then $x_{uv}^1x_{uv}^2\in M$, otherwise,
we could replace $x_{uv}^3x_{uv}^4$ by $x_{uv}^2x_{uv}^3$ in $M$
(the resulting matching is strongly maximal as $M$ is so),
contradicting the choice of $M$.

\item[$(ii)$] If the edges
$x_{uv}^3x_{uv}^4$ and $x_{uv}^1x_{uv}^2$ are in $M$, then we have
that vertices $u$ and $x^1_{vu}$ are each matched by $M$.
Otherwise, if $u$ is unmatched, we can replace $x_{uv}^1x_{uv}^2,
x_{uv}^3x_{uv}^4 \in M$ with the edges $x_{uv}^2x_{uv}^3,
ux_{uv}^1$. This again yields a strongly maximal matching (since
$u$ has no neighbors unmatched by $M$), contradicting the choice
of $M$. We can use the same argument in the case $x^1_{vu}$ is
unmatched.
\end{itemize}

This is also some of the steps in order to transform any minimum strongly maximal matching into the desired one.

Now, let $uv$ be and edge in the graph $G$ such that $x_{uv}^3x_{uv}^4\in M$. By $(i)$ and $(ii)$, we can deduce that  $|M \cap F_{uv}|=4$. Consider the matching
\[
\tilde M:=(M \setminus F_{uv}) \cup F_{uv}^{\notin}
\]
We claim that $\tilde M$ is strongly maximal. As $\tilde M$ is smaller than $M$, we thus obtain the desired contradiction.

So assume $\tilde M$ is not strongly maximal. Then, as $u$ is matched, there is a augmenting path $P$ of length $1$ or $3$ starting at $v$.

Now, observe that all neighbors of $v$ are of the form $x^1_{vw}$ (for some $w\in V(G)$), and thus, as neither $x^1_{vw}x^2_{vw}$ nor $x^1_{vw}x^2_{vw}x^3_{vw}x^4_{vw}$ is an augmenting path for the strongly maximal matching $M$, all neighbors of $v$ are matched by $M$.

 So, $P$ has length $3$, and it is easy to see that $P$ has to end in some (unmatched) vertex $w\in V(G)\setminus\{u,v\}$ (by the maximality of $M$, every vertex $x^3_{wz}$ is matched by $M$, and by the choice of $M$, every vertex $x^2_{wz}$ is matched by $M$). By $(i)$ and $(ii)$, we know that $F_{vw}\cap M=F^{\notin}_{vw}$. Consider the matching
\[
 (\tilde M \setminus F_{vw}^{\notin}) \cup F_{vw}^{\in}.
\]
This matching is clearly strongly maximal, and has fewer edges of the form $x_{uv}^3x_{uv}^4$, contradicting the choice of $M$. (And this is the remaining step in order to transform any minimum strongly maximal matching into the desired one.) This ends the proof of Claim \ref{claimNP}.\\

Therefore, by Claim \ref{claimNP}, we have that there is a minimum strongly maximal $M'$ in $H_G$ that verifies either $F_{uv} \cap M' =  F^{\in}_{uv}$ or $F_{uv} \cap M' = F^{\notin}_{uv}$ for each edge $uv \in E(G)$.

Next we show that if $M$ is a minimum maximal matching of $G$ and $M'$ is a minimum strongly maximal matching of $H_G$, $|M| = |M'|-3m$.
As explained above, this completes the proof.

Let $M$ be a minimum maximal matching of $G$. Using the auxiliary sets $F_{uv}^{\in}$ and $F_{uv}^{\notin}$, we define a strongly maximal matching $M'$ of $H_G$ by
\begin{align*}
M' = \bigcup_{uv \in M}F_{uv}^{\in}\cup \bigcup_{uv \notin M}F_{uv}^{\notin}.
\end{align*}
Note that $|M'| = |M| + 3m$.

Now, let $M'$ be a minimum strongly maximal matching of $H_G$ that verifies either $F_{uv} \cap M' =  F^{\in}_{uv}$ or $F_{uv} \cap M' = F^{\notin}_{uv}$ for each edge $uv \in E(G)$. We define a maximal matching $M$ of $G$ by setting
\[
M = \{uv : uv\in E(G), \ F_{uv}\cap M'= F^{\in}_{uv}\}.
\]
Clearly $|M| = |M'|-3m$, which completes the proof.
\end{proof}


\section{b-coloring co-trees and tree-cographs}

\subsection{co-trees}

\begin{theorem}\label{thm:algo-Xb-cotrees}
In the class of co-trees, $\chi_b$ can be computed in polynomial time.
\end{theorem}

\begin{proof}
According to Lemma~\ref{lem:matching}, the problem is equivalent to find a minimum strongly maximal matching (\textsc{msmm}) in a tree.
We will do it by dynamic programming. In order to do so, we will define five functions $F_i(r,s)$, $i=1,\dots,5$, for a nontrivial tree $T_{rs}$ rooted at a leaf $r$ with neighbor $s$.
As we will apply them to the subtrees of a tree, we will assume that $r$ can have neighbors outside $T_{rs}$.
\begin{itemize}
\item $F_1(r,s)$: cardinality of a \textsc{msmm} of $T_{rs}$ such that $r$ is unmatched, and $\infty$ if it does not exist.

\item $F_2(r,s)$: cardinality of a \textsc{msmm} of $T_{rs}$ that uses the edge $rs$ and such that $s$ may or may not have an unmatched neighbor (this case will apply when $r$ has no unmatched neighbor outside $T_{rs}$), and $\infty$ if it does not exist.

\item $F_3(r,s)$: cardinality of a \textsc{msmm} of $T_{rs}$ that uses the edge $rs$ and such that $s$ cannot have an unmatched neighbor (this case will apply when $r$ has already an unmatched neighbor outside $T_{rs}$, so an unmatched neighbor of $s$ will complete an augmenting path of length $3$ in the whole tree), and $\infty$ if it does not exist.

\item $F_4(r,s)$: cardinality of a \textsc{msmm} of $T_{rs}$ such that the vertex $s$ is matched with some vertex different from $r$ and the vertex $r$ is considered as ``already matched'' (this case will apply when $r$ is already matched with a vertex outside $T_{rs}$), and $\infty$ if it does not exist.

\item $F_5(r,s)$: cardinality of a \textsc{msmm} of $T_{rs}$ such that the vertex $s$ remains unmatched and the vertex $r$ is considered as ``already matched'', and $\infty$ if it does not exist.
\end{itemize}

With these definitions, for the base case in which $V(T_{rs})=\{r,s\}$, we have

\begin{itemize}
\item $F_1(r,s)=\infty$ (if $r$ is unmatched and $s$ has no further neighbors, the matching will never be maximal)
\item $F_2(r,s)=1$ (precisely, the edge $rs$)
\item $F_3(r,s)=1$ (precisely, the edge $rs$)
\item $F_4(r,s)=\infty$ (it is not feasible because $s$ has no further neighbors)
\item $F_5(r,s)=0$
\end{itemize}

For the case in which $s$ has children $v_1, \dots, v_k$, we have

\begin{itemize}
\item $F_1(r,s)=\min_{i=1,\dots,k} \{F_3(s,v_i)+\sum_{j=1,\dots,k;
 j \neq i} \min\{F_4(s,v_j),F_5(s,v_j)\}\}$.

 In order to obtain a
maximal matching, we need to match $s$ with some of its children,
say $v_i$. Since $r$ will be unmatched, $v_i$ should not have an
unmatched neighbor, in order to prevent an augmenting path of
length $3$. When considering the trees $T_{sv_j}$ for $j\neq i$,
the vertex $s$ will have the status of ``already matched''.
Furthermore, since we are already assuming that $s$ has an
unmatched neighbor, we do not need to care about the vertices
$v_j$ being matched or not.

\item $F_2(r,s)=1+\sum_{i=1,\dots,k}
\min\{F_4(s,v_i),F_5(s,v_i)\}$.

We will use the edge $rs$, and then when considering the trees
$T_{sv_i}$ for $i=1,\dots,k$, the vertex $s$ will have the status
of ``already matched''. Furthermore, since $s$ may or may not have
an unmatched neighbor, we can take the minimum over $F_4$ and
$F_5$ for each of the trees $T_{sv_i}$.

\item $F_3(r,s)=1 +\sum_{i=1,\dots,k} F_4(s,v_i)$.

This case is similar to the previous one, but now the vertex $s$
cannot have unmatched neighbors, so we will just consider $F_4$
for each of the trees $T_{sv_i}$.

\item $F_4(r,s)=\min \{ \min_{i=1,\dots,k}
\{F_2(s,v_i)+\sum_{j=1,\dots,k;
 j \neq i} F_4(s,v_j)\},$ \linebreak $\min_{i=1,\dots,k} \{F_3(s,v_i)+\sum_{j=1,\dots,k;
 j \neq i} \min\{F_4(s,v_j),F_5(s,v_j)\}\}\}$

As in the first case, we need to match $s$ with some of its
children, say $v_i$. But now, since $r$ is assumed to be matched,
$s$ may or may not have an unmatched neighbor, depending on the
matching status of the vertices $v_j$ with $j\neq i$. So we will
take the minimum among allowing $v_i$ to have an unmatched
neighbor and forcing $v_j$, $j\neq i$, to be matched, or
forbidding $v_i$ to have an unmatched neighbor and allowing $v_j$,
$j\neq i$, to be either matched or not.

\item $F_5(r,s)=\sum_{i=1,\dots,k} F_1(s,v_i)$

This last case is quite clear.
\end{itemize}

In this way, in order to obtain the cardinality of a minimum
strongly maximal matching of a nontrivial tree $T$, we can root it
at a leaf $r$ whose neighbor is $s$ and compute
$\min\{F_1(r,s),F_2(r,s)\}$. By keeping some extra information, we
can also obtain in polynomial time the matching itself.
\end{proof}


\subsection{Tree-cographs}

A graph is a \emph{tree-cograph} if it can be constructed from trees by disjoint union and complement operations. Tree-cographs have been introduced by Tinhofer \cite{Tinhofer89} as a generalization of trees and cographs.

Let $G_1 = (V_1,E_1)$ and $G_2 = (V_2,E_2)$ be two graphs with
$V_1 \cap V_2 = \emptyset$. The \emph{union} of $G_1$ and $G_2$ is
the graph $G_1 \cup G_2 = (V_1 \cup V_2, E_1 \cup E_2)$, and the
\textit{join} of $G_1$ and $G_2$ is the graph $G_1 \vee G_2 = (V_1
\cup V_2, E_1 \cup E_2 \cup V_1 \times V_2)$. Note that
$\overline{G_1\vee G_2} = \overline{G_1} \cup \overline{G_2}$.

 Tree-cographs can be recursively defined as follows: a graph $G$ is a tree-cograph if and only if
\begin{itemize}
\item[(i)] $G$ is a tree or a co-tree, or
\item[(ii)] $G$ is the union of two tree-cographs $G_1$ and $G_2$, or
\item[(iii)] $G$ is the join of two tree-cographs $G_1$ and $G_2$.
\end{itemize}
Notice that if $(i)$ in the above definition is replaced by ``$G$ is a single vertex" then, the obtained graph is a cograph.

The notion of \emph{dominance sequence} has been introduced in
\cite{Bonomo-et-al09} in order to compute the b-chromatic number
of $P_4$-sparse graphs and, in particular, cographs. Formally,
given a graph $G$, the dominance sequence $\mbox{dom}_G \in
\mathbb{Z}^{\mathbb{N}_\geq \chi(G)}$, is defined such that
$\mbox{dom}_G[t]$ is the maximum number of distinct color classes
admitting dominant vertices in any coloring of $G$ with $t$
colors, for every $t \geq \chi(G)$. Note that it suffices to
consider this sequence until $t = |V(G)|$, since $\mbox{dom}_G[t]
= 0$ for $t > |V(G)|$. Therefore, in the sequel we shall consider
only the \emph{dominance vector}
$(\mbox{dom}_G[\chi(G)],\ldots,\mbox{dom}_G[|V(G)|])$. Notice that
a graph $G$ admits a b-coloring with $t$ colors if and only if
$\mbox{dom}_G[t] = t$. Moreover, it is clear that
$\mbox{dom}_G[\chi(G)] = \chi(G)$.

The following results given in \cite{Bonomo-et-al09} are very
important in order to compute the b-chromatic number of graphs
that can be decomposed recursively in modules via disjoint union
or join operations.

\begin{theorem}[\cite{Bonomo-et-al09}]
\label{teo:Xb-union}
Let $G_1 = (V_1,E_1)$ and $G_2 = (V_2,E_2)$ be two graphs such
that $V_1 \cap V_2 = \emptyset$. If $G = G_1 \cup G_2$ and $t \geq
\chi(G)$, then
$$\dom_G[t]=\min\{t,\dom_{G_1}[t]+\dom_{G_2}[t]\}.$$
\end{theorem}

\begin{theorem}[\cite{Bonomo-et-al09}]
\label{teo:Xb-join}
Let $G_1 = (V_1,E_1)$ and $G_2 = (V_2,E_2)$ be two graphs such
that $V_1 \cap V_2 = \emptyset$. Let $G = G_1 \vee G_2$ and
$\chi(G) \leq t \leq |V(G)|$. Let $a = \max \{\chi(G_1),t -
|V(G_2)|\}$ and $b = \min \{|V(G_1)|,t - \chi(G_2)\}$. Then $a
\leq b$ and $$\dom_G[t]=\max_{a \leq j \leq
b}\{\dom_{G_1}[j]+\dom_{G_2}[t-j]\}.$$
\end{theorem}

In order to compute the dominance vector of a tree-cograph and its corresponding b-chromatic number, by Theorems \ref{teo:Xb-union} and \ref{teo:Xb-join}, it is sufficient to compute the dominance vector for both trees and co-trees.

\subsubsection{Dominance vector for trees}

Irving and Manlove~\cite{I-M-b-col} have shown that the b-chromatic number of any tree $T$ is equal to $m(T)-1$ or $m(T)$, depending on the existence of a unique vertex in $T$ called a \emph{pivot}, defined as follows. 

A vertex $v$ of $T$ is called \emph{dense} if $d(v) \geq m(T)-1$.
Call $T$ \emph{pivoted} if $T$ has exactly $m(T)$ dense vertices, and contains a distinguished vertex $v$, called a \emph{pivot} of $T$, such that:
(1) $v$ is not dense,
(2) each dense vertex is adjacent either to $v$ or to a dense vertex adjacent to $v$, 
and (3) any dense vertex adjacent to $v$ and to another dense vertex has degree $m(T)-1$.

Irving and Manlove~\cite{I-M-b-col} show that a pivot is unique if it exists and that pivoted trees can be recognized in linear time. Moreover, they obtain the following result.

\begin{theorem}[\cite{I-M-b-col}]
\label{theo-trees1}
Let $T$ be a tree. If $T$ is pivoted then $\chi_b(T) = m(T) - 1$, otherwise $\chi_b(T) = m(T)$. In both cases, a b-coloring of $T$ with $\chi_b(T)$ colors can be obtained in linear time.
\end{theorem}

It is known that chordal graphs are b-continuous \cite{K-K-V-b-col, Faik-tesis} and thus trees are b-continuous as well. 
Therefore, we may derive the following result concerning the dominance vector for trees.

\begin{lemma}
\label{dom-vec-trees-l1}
Let $T$ be a tree with maximum degree $\Delta$. Then, $\mbox{dom}_T[i] = i$, for $2 \leq i \leq \chi_b(T)$. Moreover, $\mbox{dom}_T[i] = 0$ for any $i > \Delta + 1$.
\end{lemma}

Moreover, it is not difficult to obtain a b-coloring of a tree $T$ with $i$ colors from one with $i+1$ colors in polynomial time, for $2 \leq i < \chi_b(T)$.

\begin{lemma}
\label{dom-vec-trees-l2}
Let $T$ be a pivoted tree. Then $\mbox{dom}_T[m(T)] = m(T) -1$.
\end{lemma}

\begin{proof}
By Theorem \ref{theo-trees1}, if $T$ is a pivoted tree then $\chi_b(T) = m(T) - 1$, so $\mbox{dom}_T[m(T)] \leq m(T) -1$. Consider now the following coloring of $T$ with $m(T)$ colors. Give color $1$ to the pivot $v$ of $T$. Since $v$ is not dense, there are at least two dense vertices at distance $2$ of $v$; give color $1$ to one of them, say $w$.
Now color the dense vertices using the $m(T)$ different colors and color their neighbors in such a way that the only dense vertex that is not dominant is the common neighbor to $v$ and $w$. It is easy to extend this coloring to a proper coloring of $T$ with $m(T)$ colors.
\end{proof}

We now show how to compute the values $\mbox{dom}_T[i]$ and a coloring of $T$ with $i$ colors and $\mbox{dom}_T[i]$ dominant vertices in linear time, for $m(T) < i \leq \Delta +1$. 
For this, we need the following definition.
Let $T$ be a tree of maximum degree $\Delta$ and let $i$ be an integer such that $m(T) < i \leq \Delta + 1$. 
We define $m_i(T)$ as the number of vertices in~$T$ of degree at least $i-1$.

It is not difficult to see that for a tree $T$, $\mbox{dom}_T[i] \leq m_i(T) < i$, for values of $i$ with $m(T) < i \leq \Delta + 1$.

\begin{lemma}
Let $T$ be a tree of maximum degree $\Delta$ and let $i$ be an integer with $m(T) < i \leq \Delta + 1$. Then, $\mbox{dom}_T[i] = m_i(T)$, and a coloring of $T$ with~$i$ colors and $m_i(T)$ dominant vertices can be computed in linear time.
\end{lemma}

\begin{proof}
For convenience, set $k = m_i(T)$. As $i\leq\Delta +1$, we have $k>0$. Let $P$ be a path disjoint from $T$ that contains the $i-k+3 \geq 4$ vertices $x,y,v_1,v_2,\ldots,v_{i-k},z$ in this order. We construct a tree~$T'$ disjoint from $T$ by taking $P$ and pending $i-3$ leaves from each vertex $v_j$, with $1 \leq j \leq i-k$. Obtain $T''$ from $T$ and $T'$ by adding an edge between $x$ and some leaf $h$ of~$T$.

By construction, $m_i(T'') = i$, and thus also $m(T'')=i$. Further,~$T'$ contains the dense vertices $v_1,v_2,\ldots,v_{i-k}$ of $T''$, and~$T$ contains $k$ dense vertices of $T''$. So, as both $x$ and $y$ have degree~$2$ in~$T''$ (and thus either both or none of them are dense), we see that $T''$ is not pivoted. Hence,  Theorem \ref{theo-trees1} yields that $\chi_b(T'') = m(T'') = i$, and a b-coloring of $T''$ with $i$ colors can be computed in linear time.

Now, the dominant vertices in $T''$ are exactly the $k$ vertices of degree at least $i-1$ in $T$ and the $i-k$ vertices $v_1,\ldots,v_{i-k}$ in $T'$. Therefore, by removing the tree $T'$ from $T''$ we obtain the desired coloring of $T$ with $i$ colors and exactly $k$ dominant vertices. Moreover, notice that the distance in $T''$ between a dense vertex in $T$ and a dense vertex in $T'$ is at least equal to $4$. Hence, by using Irving's and Manlove's algorithm \cite{I-M-b-col} for b-coloring $T''$ with $i$ colors, we can forget the tree $T'$ and thus, the coloring of $T$ with $i$ colors and $m_i(T)$ dominant vertices can be done in $O(|V(T)|)$ time.
\end{proof}

\subsubsection{Dominance vector for co-trees}

Let $G$ be a graph and $M$ be a matching of it. Let $S_1(G,M)$ be
the number of unmatched vertices that have at least an unmatched
neighbor and $S_2(G,M)$ be the number of edges of $M$ that are the
center of an augmenting path of length $3$ for $M$. Now, let
$F(\overline{G},k)$ be the minimum of $S_1(G,M)+S_2(G,M)$ over all
the matchings $M$ of $G$ with $|M|=k$.

Now, let $G$ be a graph with stability at most two and consider a
coloring of it. Let $M$ be the matching of $\overline{G}$
corresponding to that coloring. The number of color classes
without a dominant vertex are exactly
$S_1(\overline{G},M)+S_2(\overline{G},M)$. So, for $\chi(G) \leq i
\leq |V(G)|$, $\mbox{dom}_{G}[i] = i-F(\overline{G},|V(G)|-i)$. We
will show how to compute $F(T,k)$ for a tree $T$ and a nonnegative
integer $k$ in polynomial time.

\begin{theorem}\label{thm:algo-dom-cotrees}
If $G$ is a co-tree, $\mbox{dom}_G$ can be computed in polynomial time.
\end{theorem}

\begin{proof}
As we noticed above, the problem is equivalent to compute
$F(\overline{G},k)$. We will do it by dynamic programming. In
order to do so, and in a similar fashion as in Theorem
\ref{thm:algo-Xb-cotrees}, we will define seven
functions $F_i(r,s,k)$, $i=1,\dots,7$, for a nontrivial tree
$T_{rs}$ rooted at a leaf $r$ with neighbor $s$ and a nonnegative integer $k$. As we will apply
them to the subtrees of a tree, we will assume that $r$ can have
neighbors outside $T_{rs}$. Nevertheless, we will count for $S_2$ just the edges of $M \cap E(T_{rs})$ and for $S_1$ the vertices of
$V(T_{rs})$, with the exception of $r$ when it is unmatched but has already an unmatched neighbor outside $T_{rs}$, in order to avoid double counting.

For $i = 1, \dots, 7$, $F_i(r,s,k)$ will be the minimum of
$S_1(T_{rs},M)+S_2(T_{rs},M)$ over all the matchings $M$ with
$|M|=k$ such that:

\begin{itemize}
\item $F_1(r,s,k)$: $r$ is unmatched and $s$ is matched by $M$
with some vertex different from $r$.

\item $F_2(r,s,k)$: $M$ uses the edge $rs$ and $r$ has no
unmatched neighbor outside $T_{rs}$.

\item $F_3(r,s,k)$: $M$ uses the edge $rs$ and $r$ has an
unmatched neighbor outside $T_{rs}$.

\item $F_4(r,s,k)$: the vertex $s$ is matched by $M$ with some
vertex different from $r$ and the vertex $r$ is already matched
with a vertex outside $T_{rs}$.

\item $F_5(r,s,k)$: the vertex $s$ remains unmatched and the
vertex $r$ is already matched with a vertex outside $T_{rs}$.

\item $F_6(r,s,k)$: $r$ is unmatched, $s$ remains unmatched, and
$r$ has no unmatched neighbor outside $T_{rs}$.

\item $F_7(r,s,k)$: $r$ is unmatched, $s$ remains unmatched, and
$r$ has an unmatched neighbor outside $T_{rs}$ (we will not count
$r$ for $S_1$ as we assume it is already counted).
\end{itemize}

In any case, the value will be $\infty$ if no such $M$ does exist.

With these definitions, for the base case in which
$V(T_{rs})=\{r,s\}$, we have

\begin{itemize}
\item $F_1(r,s,k)=\infty$ ($s$ has no further neighbors)
\item For $i=2,3$, $F_i(r,s,1)=0$ (we define $M = \{rs\}$), $F_i(r,s,k)=\infty$ for $k \neq 1$.
\item $F_4(r,s,k)=\infty$ (it is not feasible because $s$ has no further neighbors)
\item $F_5(r,s,0)=0$, $F_5(r,s,k)=\infty$ for $k \neq 0$.
\item $F_6(r,s,0)=2$, $F_6(r,s,k)=\infty$ for $k \neq 0$.
\item $F_7(r,s,0)=1$, $F_7(r,s,k)=\infty$ for $k \neq 0$.
\end{itemize}

For the case in which $s$ has children $v_1, \dots, v_{\ell}$, we have

\begin{itemize}
\item $F_1(r,s,0)=\infty$, $F_1(r,s,k)=\min_{k_1+\dots+k_{\ell}=k}
\min_{i=1,\dots,{\ell}}$ $\{F_3(s,v_i,k_i)$ \linebreak
$+\sum_{j=1,\dots,{\ell};
 j \neq i} \min\{F_4(s,v_j,k_j),F_5(s,v_j,k_j)\}\}$ for $k>0$.

We need to match $s$ with some of its children, say $v_i$. Since
$r$ will be unmatched, $s$ will have an unmatched neighbor outside
$T_{sv_j}$. When considering the trees $T_{sv_j}$ for $j\neq i$,
the vertex $s$ will have the status of ``already matched''.
Furthermore, since we are already assuming that $s$ has an
unmatched neighbor, we do not need to distinguish about the
vertices $v_j$ being matched or not.

\item $F_2(r,s,0)=\infty$,
$F_2(r,s,k)=\min_{k_1+\dots+k_{\ell}=k-1} \sum_{i=1,\dots,{\ell}}
\min\{F_4(s,v_i,k_i),$ $F_5(s,v_i,k_i)\}$ for $k > 0$.

We will use the edge $rs$, and then when considering the trees
$T_{sv_i}$ for $i=1,\dots,{\ell}$, the vertex $s$ will have the
status of ``already matched'' and we will use $k-1$ edges in total
(thus for $k=0$ it is not feasible). Furthermore, since $r$ has no
unmatched neighbor, we can take the minimum over $F_4$ and $F_5$
for each of the trees $T_{sv_i}$ and in no case the edge $rs$ will
be the center of an augmenting path of length $3$.

\item $F_3(r,s,0)=\infty$,
$F_3(r,s,k)=\min_{k_1+\dots+k_{\ell}=k-1}
\min\{\sum_{i=1,\dots,{\ell}} F_4(s,v_i,k_i),$ $1+
\sum_{i=1,\dots,{\ell}} \min\{F_4(s,v_i,k_i),F_5(s,v_i,k_i)\}\}$
for $k > 0$.

This case is similar to the previous one, but now, since $r$ has an
unmatched neighbor, we distinguish between the case in which we consider $F_4$
for each of the trees $T_{sv_i}$ so that the edge $rs$ will not
be the center of an augmenting path of length $3$, and the case in which we take the minimum over $F_4$ and $F_5$
for each of the trees $T_{sv_i}$ and we allow the edge $rs$
being the center of an augmenting path of length $3$. In that case we will assume indeed that the edge $rs$
becomes the center of an augmenting path of length $3$, because otherwise the minimum will be attained by the previous case.

\item $F_4(r,s,0)=\infty$, $F_4(r,s,k)=\min_{k_1+\dots+k_{\ell}=k} \min \{ \min_{i=1,\dots,{\ell}}
\{F_2(s,v_i,k_i)+\sum_{j=1,\dots,{\ell};
 j \neq i} F_4(s,v_j,k_j)\},$ $\min_{i=1,\dots,{\ell}} \{F_3(s,v_i,k_i)+$ $\sum_{j=1,\dots,{\ell};
 j \neq i}$ $\min\{F_4(s,v_j,k_j),$ $F_5(s,v_j,k_j)\}\}\}$ for $k > 0$.

As in the first case, we need to match $s$ with some of its
children, say $v_i$. But now, since $r$ is assumed to be matched,
$s$ may or may not have an unmatched neighbor, depending on the
matching status of the vertices $v_j$ with $j\neq i$. So we will
take the minimum among allowing $v_i$ to have an unmatched
neighbor and forcing $v_j$, $j\neq i$, to be matched, or
forbidding $v_i$ to have an unmatched neighbor and allowing $v_j$,
$j\neq i$, to be either matched or not.

\item $F_5(r,s,k)=\min_{k_1+\dots+k_{\ell}=k} \min \{ \sum_{i=1,\dots,{\ell}} F_1(s,v_i,k_i), \min_{i=1,\dots,{\ell}}$
$\{F_6(s,v_i,k_i)+\sum_{j=1,\dots,{\ell};
 j \neq i} F_1(s,v_j,k_j)\},$ $1+\sum_{i=1,\dots,{\ell}} \min\{F_1(s,v_i,k_i),F_7(s,v_i,k_i)\}\}$

We will take the minimum over three cases: either all the $v_i$ will be matched, or exactly one of them will be unmatched,
or at least two of them will be unmatched. In this last case we count $s$ as an unmatched vertex with unmatched neighbor but
we do not force explicitly two of the $v_i$ to be unmatched, because otherwise
the minimum will be attained by one of the previous cases.

\item $F_6(r,s,k)=2+\min_{k_1+\dots+k_{\ell}=k} \sum_{i=1,\dots,{\ell}} \min\{F_1(s,v_i,k_i),F_7(s,v_i,k_i)\}$

We are counting $r$ and $s$ as unmatched vertices with an unmatched neighbor.

\item $F_7(r,s,k)=1+\min_{k_1+\dots+k_{\ell}=k} \sum_{i=1,\dots,{\ell}} \min\{F_1(s,v_i,k_i),F_7(s,v_i,k_i)\}$

We are counting just $s$ as an unmatched vertex with unmatched neighbor, since $r$ is assumed to be already counted.

\end{itemize}

Notice that as the values of the functions $F_i$ are bounded by
the number of vertices of the corresponding tree, and $k$ is also
bounded by that number, taking the minimum over
$k_1+\dots+k_{\ell}=k$ of some combination of these $F_i$ is
equivalent to solve a polynomially bounded number of knapsack
problems where both the weights and the utilities are polynomially
bounded as well, so this can be done by dynamic programming in
polynomial time \cite{Dantzig57}.

In this way, in order to obtain $F(T,k)$ for a nontrivial tree
$T$, we can root it at a leaf $r$ whose neighbor is $s$ and
compute $\min\{F_1(r,s,k),F_2(r,s,k),F_6(r,s,k)\}$. By keeping
some extra information, we can also obtain in polynomial time the
matching itself.
\end{proof}

\subsubsection{b-continuity and b-monotonicity of tree-cographs}

The following result was proved for union and join of graphs.

\begin{lemma}\label{lem:cont-union-join}\cite{Bonomo-et-al09}
Let $G_1 = (V_1,E_1)$ and $G_2 = (V_2,E_2)$ be two graphs such
that $V_1 \cap V_2 = \emptyset$. If $G_1$ and $G_2$ are
b-continuous, then $G_1 \cup G_2$ and $G_1 \vee G_2$ are b-continuous.
\end{lemma}

As a corollary of the lemma, Theorem \ref{thm:bconti}, and the b-continuity of chordal graphs \cite{K-K-V-b-col,Faik-tesis}, we have the following result.

\begin{theorem}\label{thm:b-cont-tree-cographs}
Tree-cographs are b-continuous.
\end{theorem}

Concerning the b-monotonicity, the following results are known for general graphs and for union and join of graphs.

\begin{lemma}\label{lem:maximum}\cite{Bonomo-et-al09}
Let $G$ be a graph. The maximum value of $\dom_G[t]$ is attained
in $t = \chi_b(G)$.
\end{lemma}

\begin{lemma}\label{lem:monot-union}\cite{Bonomo-et-al09}
Let $G_1 = (V_1,E_1)$ and $G_2 = (V_2,E_2)$ be two graphs such
that $V_1 \cap V_2 = \emptyset$, and let $G = G_1 \cup G_2$.
Assume that for every $t\ge\chi(G_i)$ and every induced subgraph
$H$ of $G_i$ we have $\dom_H[t] \leq \dom_{G_i}[t]$, for $i=1,2$.
Then, for every $t \geq \chi(G)$ and every induced subgraph $H$ of
$G$, $\dom_H[t] \leq \dom_{G}[t]$ holds.
\end{lemma}

\begin{lemma}\label{lem:monot-join}\cite{Bonomo-et-al09}
Let $G_1 = (V_1,E_1)$ and $G_2 = (V_2,E_2)$ be two b-continuous
graphs such that $V_1 \cap V_2 = \emptyset$, and let $G = G_1 \vee
G_2$. Assume that for every $t \geq \chi(G_i)$ and for every
induced subgraph $H$ of $G_i$ we have $\dom_H[t] \leq
\dom_{G_i}[t]$, for $i=1,2$. Then, for every $t \geq \chi(G)$ and
for every induced subgraph $H$ of $G$, $\dom_H[t] \leq
\dom_{G}[t]$ holds.
\end{lemma}

In order to prove the b-monotonicity of tree-cographs, we need the following two lemmas.

\begin{lemma}\label{lem:mon-dom-trees}
Let $T$ be a tree and $H$ an induced subgraph of $T$. Then for every $t \geq 2$, $\dom_H[t] \leq
\dom_T[t]$.
\end{lemma}

\begin{proof}
It is clear that it holds for $t \leq \chi_b(T)$. If $T$ is a pivoted tree, then either $m(H) < m(T)$ or the connected component of $H$ containing the dense vertices is a pivoted tree as well.
In any case, $\dom_H[m(T)] \leq
\dom_T[m(T)]$. For $t > m(T)$, let $T_j$, $j=1,\dots,k$, be the connected components of $H$. It is clear that $\sum_{j=1,\dots,k} m_t(T_j) \leq m_t(T)$, so by Theorem \ref{teo:Xb-union}, $\dom_H[t] \leq
\dom_T[t]$.
\end{proof}

\begin{lemma}\label{lem:mon-dom-cotrees}
Let $G$ be a graph with stability at most two and $H$ an induced subgraph of $G$. Then for every $t \geq \chi(G)$, $\dom_{H}[t] \leq
\dom_{G}[t]$.
\end{lemma}

\begin{proof}

The class of graphs of stability at most two is closed under
taking induced subgraphs. Thus we only have to prove that
$\dom_{G}[t]$ for a fixed $t$ is monotonously decreasing under the
deletion of a vertex. Let $H=G-v$ for some vertex $v$ of $G$. It
is clear that $\dom_{H}[t] \leq \dom_{G}[t]$ for $\chi(G) \leq t
\leq \chi_b(G)$ and for $t > |V(H)| = |V(G)|-1$. For $|V(H)| \geq
t > \chi_b(G)$, as we observed before, since $G$ and $H$ have
stability at most two, $\mbox{dom}_{G}[t] =
t-F(\overline{G},|V(G)|-t)$, and $\mbox{dom}_{H}[t] =
t-F(\overline{H},|V(H)|-t)$, where $F(X,k)$ stands for the minimum
sum of the number of unmatched vertices that have at least an
unmatched neighbor and the number of edges of $M$ that are the
center of an augmenting path of length $3$ for $M$, over all the
matchings $M$ of a graph $X$ with $|M|=k$. Then $\mbox{dom}_{G}[t]
\geq \mbox{dom}_{H}[t]$ if and only if $F(\overline{G},|V(G)|-t) =
F(\overline{G},|V(H)|+1-t) \leq F(\overline{H},|V(H)|-t)$.

Let $M$ be a matching of $\overline{H}$ that realizes this
minimum, and consider $M$ as a matching of $\overline{G}$. We need
to find a matching $M'$ of $\overline{G}$ with $|M'|=|M|+1$, which
is always posible, since $t > \chi(G)$ and then $|V(G)|-t$ is
strictly smaller than the size of a maximum matching of
$\overline{G}$.

We will consider now three cases. If $v$ has an unmatched neighbor
$w$, then let $M' = M \cup \{vw\}$. In this way,
$S_1(\overline{G},M') \leq S_1(\overline{H},M)$ and no edge of $M$
becomes the center of an augmenting path of length $3$ for $M'$ in
$\overline{G}$. Moreover, if $vw$ is the center of an augmenting
path of length $3$ for $M'$ in $\overline{G}$, then $w$ was for
$M$ an unmatched vertex of $\overline{H}$ having an unmatched
neighbor. In this case, $S_2(\overline{G},M') =
S_2(\overline{H},M)+1$ but $S_1(\overline{G},M') \leq
S_1(\overline{H},M)-1$. In any case, $F(\overline{G},|V(G)|-t)
\leq S_1(\overline{G},M')+S_2(\overline{G},M') \leq
S_1(\overline{H},M)+S_2(\overline{H},M) =
F(\overline{H},|V(H)|-t)$.

If $v$ has no unmatched neighbor but it is the end of an augmenting
path of length $3$ for $M$ in $\overline{G}$, say $vxyw$, let
$M'=M\setminus \{xy\} \cup \{vx,yw\}$. There are no new unmatched
vertices, so $S_1(\overline{G},M') \leq S_1(\overline{H},M)$ and
no edge of $M$ becomes the center of an augmenting path of length
$3$ for $M'$ in $\overline{G}$. Neither does $vx$, since $v$ had
no unmatched neighbor. If $yw$ is the center of an augmenting path
of length $3$ for $M'$ in $\overline{G}$, then $w$ was for $M$ an
unmatched vertex of $\overline{H}$ having an unmatched neighbor.
In this case $S_2(\overline{G},M') = S_2(\overline{H},M)+1$ but
$S_1(\overline{G},M') \leq S_1(\overline{H},M)-1$, so we are done.

Finally, if $v$ does not have an unmatched neighbor and it is not
the end of an augmenting path of length $3$ for $M$ in
$\overline{G}$, i.e., $S_1(\overline{G},M) = S_1(\overline{H},M)$
and $S_2(\overline{G},M) = S_2(\overline{H},M)$, let $P$ be a
minimum length augmenting path in $\overline{G}$ with respect to
$M$, and let $M' = (M \setminus E(P)) \cup (E(P) \setminus M)$.
There are no new unmatched vertices, so $S_1(\overline{G},M') \leq
S_1(\overline{H},M)$ and no edge of $M$ becomes the center of an
augmenting path of length $3$ for $M'$ in $\overline{G}$. If $P$
is of length $1$, then $S_2(\overline{G},M') \leq
S_2(\overline{H},M)+1$ but $S_1(\overline{G},M') \leq
S_1(\overline{H},M)-2$, so we are done. If $P$ is of length $3$,
we are eliminating an augmenting path of length $3$ and no new
edge becomes the center of an augmenting path of length $3$,
because otherwise there were couples of adjacent unmatched
vertices and we are supposing $P$ is of minimum length. Thus
$S_1(\overline{G},M') = S_1(\overline{H},M) = 0$ and
$S_2(\overline{G},M') \leq S_2(\overline{H},M)-1$, so we are done.
If there were no augmenting paths of length $1$ or $3$, $M$ is a
strongly maximal matching of $\overline{G}$ and then
$|V(G)|-|V(H)|+t = t+1 \leq \chi_b(\overline{G})$, a contradiction
because we were supposing $t > \chi_b(\overline{G})$. \end{proof}

So, we can conclude the following.

\begin{theorem}
Tree-cographs are b-monotonic.
\end{theorem}

\begin{proof} As tree-cographs are hereditary, it is enough to prove
that given a tree-cograph $G$, $\chi_b(G) \geq \chi_b(H)$, for every
induced subgraph $H$ of $G$. By the decomposition structure of tree-cographs \cite{Tinhofer89} and
Theorem~\ref{thm:b-cont-tree-cographs}, Lemmas~\ref{lem:monot-union}, \ref{lem:monot-join}, \ref{lem:mon-dom-trees}, and \ref{lem:mon-dom-cotrees}, an induction argument shows that for
every tree-cograph $G$, every $t \geq \chi(G)$, and every induced
subgraph $H$ of $G$, $\dom_H[t] \leq \dom_{G}[t]$ holds. Let $G$
be a tree-cograph, and let $H$ be an induced subgraph of $G$. If
$\chi_b(H) < \chi(G)$, then $\chi_b(H) < \chi_b(G)$. Otherwise,
$\chi_b(H) = \dom_H[\chi_b(H)] \leq \dom_{G}[\chi_b(H)]$, and by
Lemma~\ref{lem:maximum} $\dom_{G}[\chi_b(H)] \leq
\dom_{G}[\chi_b(G)] = \chi_b(G)$. Hence $\chi_b(G) \geq
\chi_b(H)$.\end{proof}


\bibliographystyle{plain}

\end{document}